\documentclass[]{interact}

\usepackage{graphicx,amssymb,amsmath}
\usepackage{amsthm}
\usepackage{multirow}
\usepackage{hyperref}
\usepackage{lineno}
\usepackage{xcolor}
\usepackage{colortbl}
\usepackage{comment}
\usepackage{amssymb}
\usepackage{subcaption}
\usepackage{gensymb}
\usepackage{soul}

\newcommand{\SP}{\mathcal{S}_{\mathcal{P}}}

\newcommand{\PP}{\mathcal{P}}

\newcommand{\vset}[1]{\perp_\mathcal{#1}}

\newtheorem{problem}{Problem}

\newtheorem{remark}{Remark}
\newtheorem{definition}{Definition}
\newtheorem{theorem}{Theorem}
\newtheorem{lemma}[theorem]{Lemma}
\newtheorem{corollary}[theorem]{Corollary}

\title{Modeling energy collection with shortest paths in rectangular grids: an efficient algorithm for energy harvesting}
\author{ 
\name{José-Miguel Díaz-Bañez\textsuperscript{a} and José Manuel Higes\textsuperscript{a} and Miguel-Angel Pérez-Cutiño\textsuperscript{a,b}\thanks{CONTACT M.A. Pérez-Cutiño. Email: m.perez@virtualmech.com} and Tom Todtenhaupt\textsuperscript{b}
}
\affil{\textsuperscript{a}Department of Applied Mathematics II, University of Seville, Seville, Spain; \textsuperscript{b}Virtualmechanics S.L, Seville, Spain}
}

\date{}

\begin{document}

\maketitle

\begin{abstract}
Parabolic Trough solar fields are among the most prominent methods for harnessing solar energy. However, continuous sun-tracking movements leads to wear and degradation of the tracking system, raising the question of whether the rotations can be minimized without compromising energy capture. In this paper, we address this question by exploring two problems: (1) minimizing the number of SCA rotational movements while maintaining energy production within a specified range, and (2) maximizing energy capture when the number of rotations is limited. 
Unlike prior work, we develop a general framework that considers variable conditions. By transforming the problem into grid-based path optimization, we design polynomial-time algorithms that can operate independently of the weather throughout the day.
Through realistic simulations and experiments using real-world data, our methods show that rotational movements of solar trackers can be reduced by at least 10\% while maintaining over 95\% energy collection efficiency. These results offer a scalable solution for improving the operational lifespan of the solar field. Furthermore, our methods can be integrated with solar irradiance forecasting, enhancing their robustness and suitability for real-world deployment.
\end{abstract}

\begin{keywords}
OR in Energy; Combinatorial Optimization; Solar tracking; Rectangular grids; CSP plants.
\end{keywords}

\section{Introduction}
\label{sec:intro}

The production of green energy has gained increasing attention in recent years. Many variables are interlaced in the energy production process, energy collection being one of the most important. In the context of solar energy, solutions range from static photovoltaic panels to fields of mirrors arranged to concentrate the solar energy up to some specific location; the latter are known as Concentrated Solar Power (CSP) plants. In this work we focus our attention on one particular type of CSP plant: Parabolic Trough (PT) solar fields. 

PT plants are composed of a set of parabolic-shaped mirrors or Solar Collector Elements (SCE) concentrating the solar energy into cylindrical pipes, or Heat Collector Element (HCE). A set of SCEs (typically of size twelve) is known as Solar Collector Assembly (SCA), and they are designed to track the sun with the use of robotic arms; see Figure \ref{fig:bj_good_bad}(a). Moving heavy mechanical structures such as the SCAs over long periods of time results in the deterioration of the robotic arms which at some point are not able to expand or rotate, being the root cause of catastrophic events; see Figure \ref{fig:bj_good_bad}(b). To avoid this, our work focuses on answering the following question: %, termed as question of interest (\textbf{QI}): 
\textit{Is it possible to reduce the rotations of a sun-tracking system while maintaining high energy collection?}
%can the rotations of a system tracking the sun be reduced while collecting high amounts of energy?}

\begin{figure}
    \centering
    \begin{subfigure}{.48\textwidth}
    \centering
        \includegraphics[width=0.95\linewidth, height=5cm]{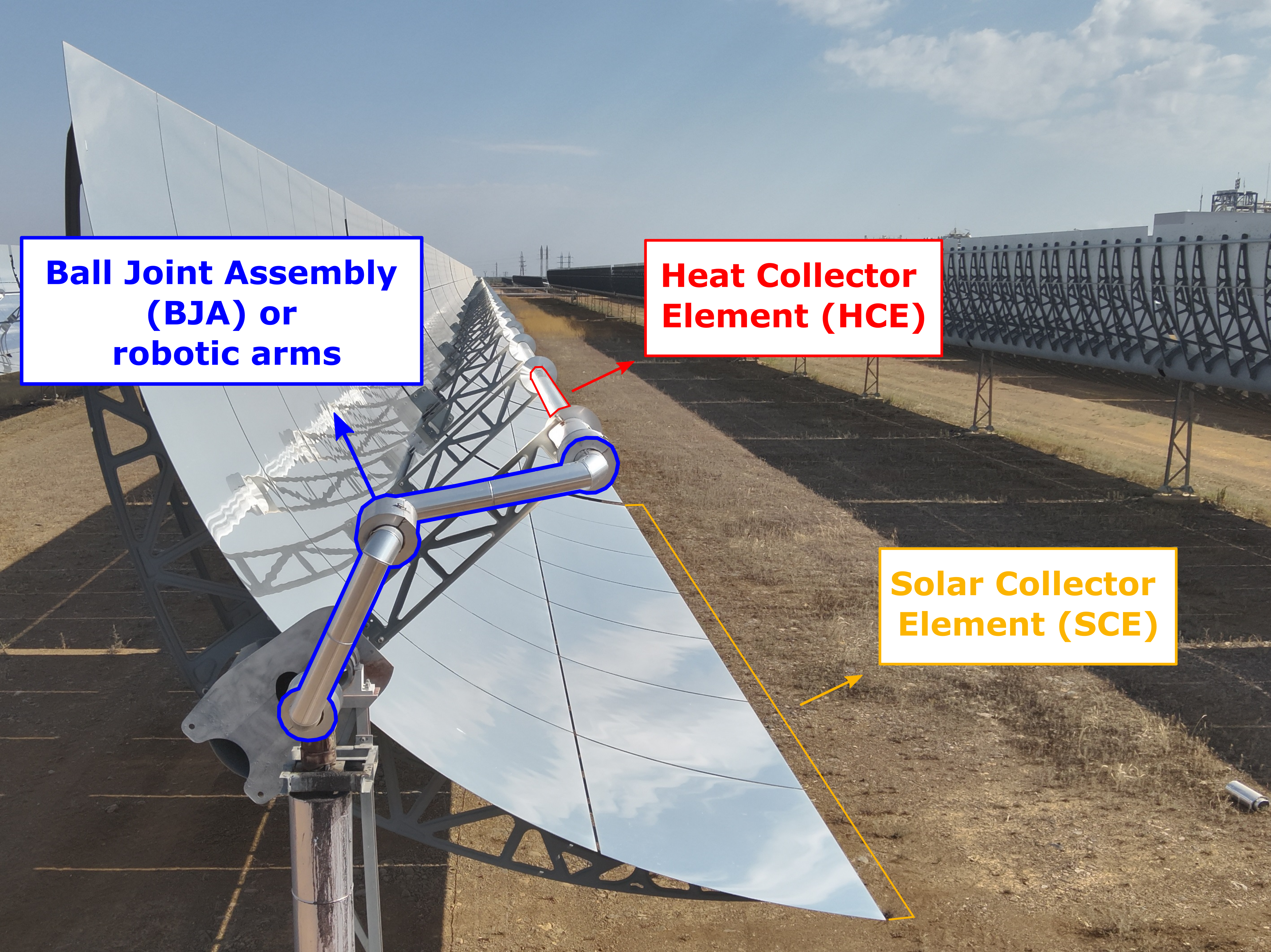}
        \caption{}
    \end{subfigure}
    \hfill
    \begin{subfigure}{.48\textwidth}
    \centering
        \includegraphics[width=0.95\linewidth, height=5cm]{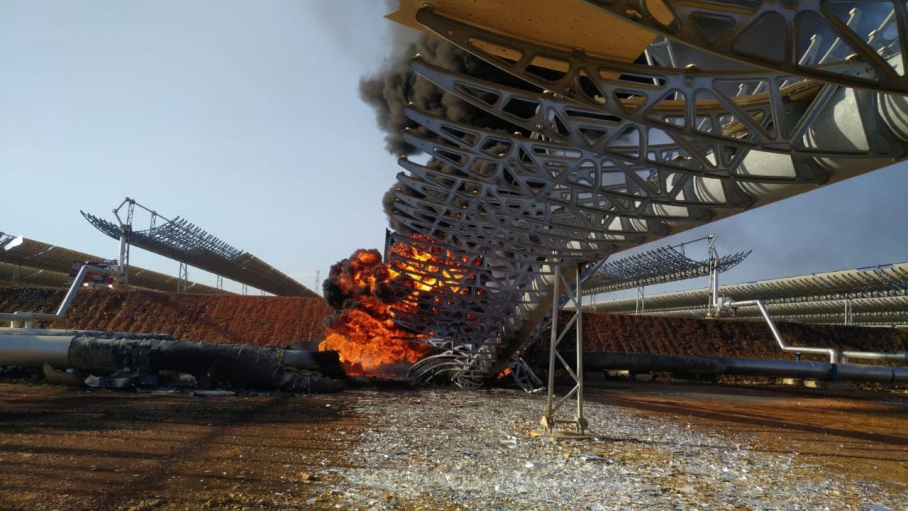}
        \caption{}
    \end{subfigure}
    \caption{(a) Main elements in a Parabolic Trough plant (b) Toxic smoke produced in an accident caused by Ball Joint Assemblies in poor condition [image extracted from \cite{noticia2022Casablanca}].}
    \label{fig:bj_good_bad}
\end{figure}

The literature concerning the problem of reducing the number of rotations of a solar tracking system (STS) is scarce. Most of the existing works concentrates on perfect tracking, which is achieved by aligning the sun with the STS by using solar equations \cite{fernandez2017mathematical} or solar sensors \cite{sabir2016optimal,sneineh2019design}. Perfect tracking is mandatory for optimal energy collection, but it discards the effect of moving heavy structures for long periods of time. The work of \cite{diaz2024optimal} proposes a solution to this problem. %e \textbf{QI}. 
However, in their research the initial tracking conditions and the azimuth angle are considered constant throughout the day, which is not suitable for practical applications. Additionally, they studied some problems for which non-polynomial time algorithms are known. In this work, we demonstrate that our target question can be answered in polynomial time independently of the weather conditions or the azimuth angle. In fact, we proof that scheduling a solar tracker is equivalent to designing optimal routes in rectangular grids if the solar irradiance function is known beforehand. Reducing continuous problems to structured grids or to rectangular areas is a technique widely employed in the literature; see for instance \cite{keshavarz2023finding} and \cite{mari2024shortest}. 
Our idea is supported by the fact that, in PT plants, solar rays are concentrated onto a receiver with radius $r>0$, which implies the existence of a short time window (acceptance angle) during which energy collection remains approximately constant; see Figure \ref{fig:sce_2d_cut}.
Moreover, solar irradiance in specific locations can be predicted with the use of forecasting methods %\citealp{bouquet2024ai}. These methods 
that can be coupled with our algorithms to provide a safe solar tracking.

We aim to address two key challenges in energy capture for Parabolic Trough plants: 1) minimizing the number of rotational movements of the SCA while maintaining the production within a given range, and 2) maximizing the energy captured by a STS when the number of rotations is upper bounded. Our main contributions can be summarized as:
\begin{itemize}
  \item    We demonstrate that any optimal solution to problems 1 and 2, in a general setting, can be transformed into an equivalent version restricted to a rectangular grid.
  Thus, a contribution of the paper is the geometrical modeling of an interesting problem for the industry.
    \item We provide optimal solutions for problem 1 and 2 independently of the weather conditions or the azimuth angle during the day. Our algorithms run in polynomial time with respect to the grid size.
    \item We present extensive experiments with the use of realistic simulations and real world data. Our results evidence that the number of rotations of solar trackers in PT plants can be reduced at least by a $10\%$ while keeping the energy collection above the $95\%$. Our methods are robust and can be coupled with forecasting methods that predicts only a few hours ahead.
\end{itemize}

The rest of the paper is organized as follows: Section \ref{sec:rw} comprises the main results in the literature concerning optimal energy collection in Concentrated Solar Power plants, Section \ref{sec:problem_definition} formally introduces our target problems, including their equivalence to path optimization in rectangular grids, Section \ref{sec:algorithms} details the algorithms designed for solving these problems, and Section \ref{sec:experiments} includes our experimental validation. Finally, Section \ref{sec:conclusions} summarizes the main results of our research and proposes future research perspectives.

\begin{figure}
    \centering
    \includegraphics[width=0.75\linewidth,page=2]{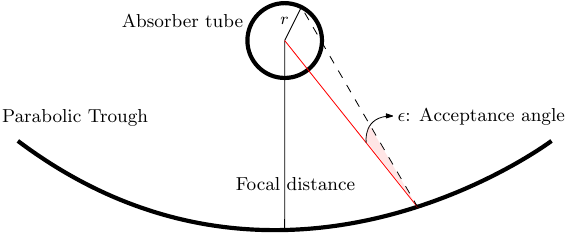}
    \caption{Two dimensional cut of a solar collector element in a Parabolic Trough solar field.}
    \label{fig:sce_2d_cut}
\end{figure}

\section{Related Work}
\label{sec:rw}

Optimizing energy collection is crucial for maximizing energy production. Numerous studies have addressed various optimization challenges related to this goal, with particular emphasis on CSP plants.

The cost of building a Solar Power Tower (SPT) plant, as well as the energy collected when distributing the heliostats in the solar field were optimized by  \cite{carrizosa2015heuristic} with greedy-based heuristics. For SPT plants, \cite{ashley2017optimisation} proposes an integer linear programming solution to the problem of optimizing the aiming point of the heliostats. In this research, it is considered that inappropriate aiming strategies can damage the central tower receiver. However, 
their target problem
does not involve moving heavy mechanical structures, hence the optimization focuses on achieving an even flux distribution across the surface of the receiver. In a similar problem, \cite{zeng2022real} corrects the position of the heliostats using Reinforcement Learning, which enables real time optimization and can be applied under varying solar conditions. Other Deep Learning-based strategies have been applied to the same problem \cite{wu2024method}. For uniform flux distribution in Linear Fresnel solar fields, optimal aiming strategies have also been designed using genetic algorithms \cite{qiu2017aiming}. Recently, \cite{andres2025solar} have introduced a collection of optimization problems based on steady-flow operation in SPT plants to benchmark blackbox solvers.
 
Energy production can be increased with the use of Model Predictive Control (MPC) for adjusting the temperature of the Heat Transfer Fluid (HTF) in Parabolic Trough solar fields \cite{gholaminejad2022stable,masero2023fast}. Regulating the fluid temperatures offers the benefit of increased robustness against disruptions, such as fluctuations in solar irradiance, or variations in mirror reflectivity \cite{camacho2012control}. In PT plants with thermal storage, MPCs have also been considered for scheduling the times where the energy is sold to the grid \cite{velarde2023scenario}.
In a similar problem, but focusing on 
photovoltaic systems with energy storage
\cite{jakoplic2021benefits} consider short forecasting methods to predict the expected power production based on ground cameras.
The aforementioned techniques are not dependent of the rotations of the system, hence they can be coupled with different tracking strategies, such as the ones described in this work. For a general review on mathematical optimization applied to optimal operation of solar thermal plants, see the work of \cite{untrau2022analysis}.

An area of active research is the design of tracking mechanisms. Depending on the number of tracking axes, the Solar Tracking System is classified as a single-axis or a dual-axis tracker. The later has the advantage of capturing more solar radiation; hence several efforts have been devoted to improve such a system \cite{tina2013intelligent,sabir2016optimal,pirayawaraporn2023innovative}. However, dual-axis trackers involves more complicated designs and higher cost which prevents their use in practical applications \cite{kong2020optimal}. Indeed, in commercial PT plants single-axis trackers are commonly used; see examples in \cite{national2003assessment}. For single-axis tracking, conventional control algorithms focus on a perfect alignment of the sun with the tracking system for maximizing the collection of direct irradiance. This is usually achieved with the use of solar equations \cite{fernandez2017mathematical} or solar sensors \cite{sabir2016optimal,sneineh2019design}.  Other control algorithms integrate the collection of diffuse irradiance \cite{anderson2022single} or consider the impact of partial shading on consecutive rows \cite{panico1991backtracking,gomez2020analysis}. Although perfect tracking is important for energy collection, the abuse of rotational moves of heavy mechanical structures can hinder its components, which is not considered in the previous works. Moreover, PT plants operates with restriction in the total temperature than can be reached by the HTF. When this temperature is too high, the STS is manually set to an out-of-focus state to prevent overheating. Instead, solar tracker considering the rotational movements of the system, or the total energy that needs to be captured might help in planning better production strategies. % The present research aims to solve optimization problems related to this type of trackers. 

The work of \cite{diaz2024optimal} proposes a solution to our target problems, but under the assumption that the total irradiance remains constant throughout the day.
This idea allows to reduce one dimension of the problems, but is only valid for limited scenarios. In our work, we consider any weather condition, which is more suitable for practical applications. Moreover, our experiments use real meteorological data from 2019, and include the adaptation of our algorithms to short-term forecasting.

\section{Problem Definition}
\label{sec:problem_definition}

Our target problems are defined in Section \ref{sec:intro}. To articulate these general statements more formally, we begin by defining the representation of the irradiance function. Subsequently, we show how to solve our problems through %optimization problems
combinatorial optimization
on rectangular grids.

\subsection{Optimization problems}

Let $f: X\times Y\rightarrow \mathbb{R}$ be a function representing the solar irradiance captured by the collector system depending on the angular value $x\in X$ of the SCA, and the angular value $y\in Y$ of the sun, both with respect to the horizontal plane.
From a practical constraint, we get 
$X = Y=[0, 180]$;
see Figure \ref{fig:f_example}(a) for an example of $f$ restricted to a smaller interval.

\begin{figure}
    \centering
    \begin{subfigure}{0.46\textwidth}
        \includegraphics[width=\textwidth]{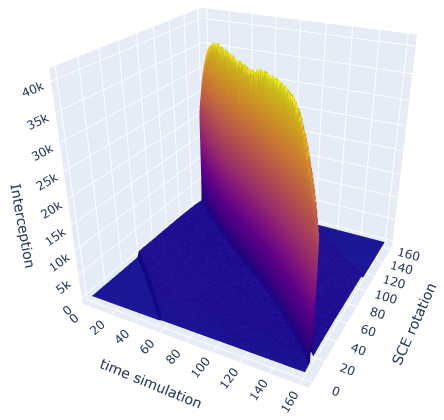}
        \caption{}
    \end{subfigure}
    \hfill
    \begin{subfigure}{0.42\textwidth}
        \includegraphics[width=\textwidth]{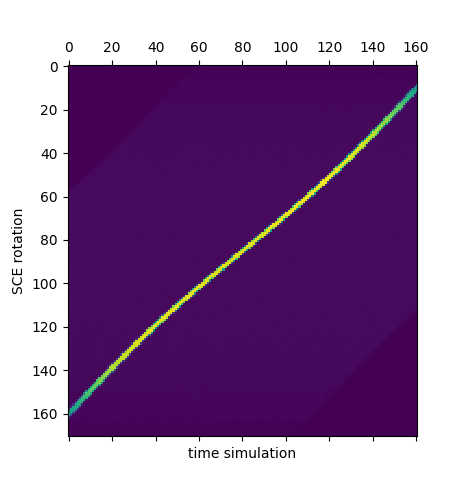}
        \caption{}
    \end{subfigure}
    \caption{Example of the irradiance function. (a) 3D representation, (b) matrix representation. }
    \label{fig:f_example}
\end{figure}

\begin{definition}
\label{def:GCP}
Given a solar irradiance function $f: X\times Y\rightarrow \mathbb{R}$, we will say that a sequence of points 
$\{P_i= (p_i, t_i)\}_{i = 1}^m \subset X\times Y$ is a generalized collector path (GCP) if: 
\begin{itemize}
\item[a.] $P_1 = (0, 0)$ and $t_m = 180$.
\item[b.] For every $1 < i < m$:
\begin{enumerate}
    \item $P_i \ne P_{i-1}$ and $P_i \ne P_{i+1}$.
    \item $p_i = p_{i+1}$ or $p_i = p_{i-1}$ but not both.
    \item $t_i = t_{i+1}$ or $t_i = t_{i-1}$ but not both.
    \item $t_i\ge t_{i-1}$.
\end{enumerate}
\end{itemize} 
An increasing collector path (ICP) is defined as a GCP that also satisfies $p_{i} \le p_{i+1}$ for every $1 \le i < m$. This means that the path does not turn to the left. The trajectory of a path is determined by its sequence of points, which specify positional and directional transitions. \par
\end{definition}

Any GCP, $\PP$, can be visualized as a snake-shaped path such that no decreasing movements in the $y$ coordinate are allowed; see Figure \ref{fig:examplepath} (left). This definition is supported by the practical application: 1) the $y$ coordinate represents the sun displacement, hence it cannot decrease; 2) 
the SCA rotates in discrete time steps,
hence when the SCA is static (moving) the corresponding portion of $\PP$ is vertical (horizontal). On the other hand, in an ICP no left movements would be allowed,
implying that the SCA only rotates forward. Notice that when $t_i = t_{i+1}$ the collector has moved from angle $p_i$ to angle $p_{i+1}$, while if $p_i = p_{i+1}$, then the collector is static collecting energy during the period of time $[t_i, t_{i+1})$. The set of points of $X \times Y$ in which the collector is static will be denoted by $\SP$. Note that $\SP$ is composed by % the set of
all
points in vertical segments of $\PP$, excluding the upper point of each segment. 

\begin{definition}\label{Definition: Energy collected}
Given a GCP, $\PP = \{(p_i, t_i)\}_{i = 1}^m$, the energy collected by $\PP$, $E(\PP)$, is defined as:
\begin{equation}
\notag
E(\PP) = \sum_{i = 1}^m \int_{t_i}^{t_{i+1}} f(p_i, y)dy
\end{equation}
\end{definition}

\begin{remark}
Notice that integrals of Definition \ref{Definition: Energy collected} are zero when the collector is not static.
\end{remark}

Given a GCP $\PP$, we could define the number $m_{\PP}$ of movements of $\PP$ as the number of times such that $p_{i} \ne p_{i+1}$. By construction of $\PP$, minimizing $m_{\PP}$ is equivalent to minimizing the number $m$ of waypoints of $\PP$, since either % $m = 2m_{\PP}$ or $m = 2m_{\PP} +1$. 
$m = 2m_{\PP} + 1$ or $m = 2m_{\PP} +2$. 
Therefore,
the problems addressed in this paper are the following: 
\begin{problem}\label{Problem: Main Problem 1}(3D Minimum Tracking Motion)
Given an irradiance function $f$, and two positive real numbers $u_1 < u_2$, find a GCP, $\PP$, such that the number of waypoints
of $\PP$ is minimum and 
for every $P_i = (p_i, t_i) \in \SP$ with $P_{i+1}=(p_i, t_{i+1})\in \PP$ $u_1 \le f((p_i, y)) \le u_2$ for every $y\in [t_i, t_{i+1})$. 
\end{problem}

\begin{problem}\label{Problem: Main Problem 2}(3D Maximum Energy Collection)
Given an irradiance function $f$, and a positive integer number $m$,   find a GCP, $\PP$, 
that maximizes the energy collection $E(\PP)$, with the constraint that the number of waypoints
does not exceed $m$.
\end{problem}

\begin{figure}
    \centering
\includegraphics[width=.95\textwidth]{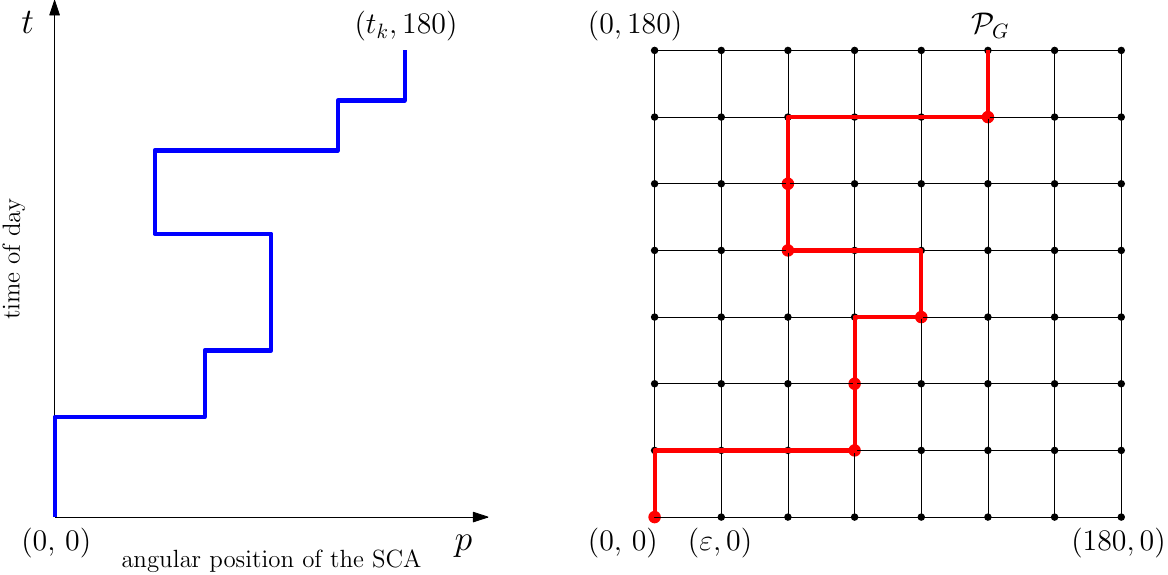}
    \caption{Left: in blue an example of a GCP. Right: in red a $\PP_G$ with %nine
    eight
    turn points. Red points  are vertical points (see Definition \ref{def:turn_and_vertical_points}). }
    \label{fig:examplepath}
\end{figure}

\subsection{Equivalence with grid graph problems}
In this section, we will transform the above optimization problems into optimal trajectory problems on a grid graph.

We can notice that the rays impacting the mirrors of the Solar Collector Assembly (SCA) are reflected to the absorber tube, which can be modeled as a cylinder of radius $r>0$. 
If we consider a cross-section, the SCA forms a perfect parabola, the absorber tube appears as a circle, and the sun can be treated as a point source due to its vast distance from Earth.
Since the rays impact with a tube and not with a point,
there is a short period of time $\varepsilon>0$ for which: 1) for a fixed position of the SCA the sun moves without affecting the concentration of rays in the absorber tube; 2) for a fixed position of the sun the SCA moves obtaining the same effect described before; see Figure \ref{fig:sce_2d_cut}. For simplicity of notation, we assume $\varepsilon$ to be the same for situations 1) and 2). The results described herein are not affected by this assumption. 

From the previous paragraph, we conclude that the domain $X \times Y$ can be decomposed into cells of size $\varepsilon \times \varepsilon$, where the irradiance function $f$ can be considered constant within each of these cells.
More formally, assume that $\varepsilon$ is computed (note that it will be related to the radius $r$ of the cilinder). Divide $X = [0, 180)$ in  $n$ subintervals $[x_i, x_{i+1})$ of size $\varepsilon$, and similarly divide $Y= [0, 180)$ in $n$ subintervals $[y_i, y_{i+1})$ of size $\varepsilon$. For each integer pair $(i, j)$ with $1\le i, j \le n$, the irradiance function satisfies $f(x, y) = f(x_i, y_j)$ if $(x, y) \in [x_i, x_{i+1}) \times [y_j, y_{j+1})$.

Given the previous decomposition of $X$ and $Y$ in $n$ subintervals of the form $[x_i, x_{i+1})$ and $[y_i, y_{i+1})$ respectively,
the \emph{weighted grid graph} $G = \langle V, E, W\rangle$ associated to the SCA
is the graph such that $V = \{(x_i, y_j): 
%\text{ for every } 
1 \le i, j \le n\}$,  $E$ is the set of vertical or horizontal segments that connect two adjacent points of $V$, and $W$ is the function $W: V \rightarrow \mathbb{R}$  defined as $W((x_i, y_j)) = f(x_i, y_j)$. For simplicity we will call the vertex as $v_{i, j} = (x_i, y_j)$ and the weights $w_{i, j} = f(x_i, y_j)$. 

\begin{definition}
Given a weighted grid graph $G$ associated to the radius of the SCA, we define a restricted generalized collector path as a GCP, $\PP_G$, such that every $(p_i, t_i)$ of $\PP_G$ is a vertex of $G$. Analogously, we can define a restricted increasing collector path from a given ICP.
\end{definition}
\begin{definition}
\label{def:turn_and_vertical_points}
    Let $\PP_G$ be a restricted generalized collector path,  and $\mathcal{S} = \{[P_i, P_{i+1})\}$ be the set of its vertical segments. We define the set of vertical points of $\PP_G$ as $\vset{\PP_G} = \bigcup [P_i, P_{i+1}) \cap V(G)$. Notice that $\vset{\PP_G}$ contains all the vertices of $G$ that lie along the vertical segments of $\PP_G$, except for the end point of each vertical segment. 
\end{definition}

Figure \ref{fig:examplepath} (right) shows an example of a restricted collector path $\PP_G$ and its vertical points (the points doted in red). Notice that vertical points lie within vertical segments of $\PP_G$, but only the start and the end points of the segments belong to the path.

We are interested in reducing our problems to these ones:
\begin{problem}
\label{problem:min_mov}
    (3D Minimum Tracking Motion for grids, or 3D-MTM): Given two real numbers $u_1, u_2$ and a weigthed grid graph $G$, find a $\PP_G$ such that: 
    \begin{itemize}
        \item For every $v_{i,j}\in\vset{\PP_G}$, $u_1\le w_{i,j}\le u_2$.
        \item $\PP_G$ is of minimum cardinality.
    \end{itemize}
\end{problem}

\begin{problem}
\label{problem:max_energy}
    (3D Maximal Energy Collection for grids, or 3D-MEC): Given an integer $m$ and a weigthed grid graph $G$, find a path $\PP_G$ such that:
    \begin{itemize}
        \item The cardinality of $\PP_G$ is at most $m$.
        \item $\sum_{v_{i,j}\in \vset{\PP_G}}w_{i,j}$ is maximal.
    \end{itemize}
\end{problem}

\begin{lemma}\label{Lemma: second energy simplification}
     Let $\PP_G$ be a restricted generalized collector path. Then: 
     
     \begin{equation}
         \notag
         E(\PP_G) = \varepsilon \cdot \sum_{v_{i,j} \in \vset{\PP_G}} w_{i,j}
     \end{equation}
\end{lemma}
\begin{proof}
It follows trivially from the definition of $E(\PP_G)$ 
and the fact that the irradiance function $f$ is constant in each cell.  
\end{proof}

\begin{lemma}\label{lemma: vertical simplification}
     Let $\PP_G$ be a restricted generalized collector path and let $u_1, u_2$ be two real numbers. Let %$\SP$ 
     $\mathcal{S}_{\mathcal{P}_G}$ be the set of static points of $\PP_G$. Every %$P \in \SP$
     $P\in \mathcal{S}_{\mathcal{P}_G}$
     satisfies $ u_1\le f(P)\le u_2$ if and only if every $P' \in \vset{\PP_G}$ satisfies $ u_1\le w_{P'}\le u_2$. 
\end{lemma}
\begin{proof}
This result follows directly from the definitions and the fact that the irradiance function remains constant within each cell.
\end{proof}

\begin{corollary}
Given a weighted grid graph $G$, then solving problem \ref{Problem: Main Problem 1} (\ref{Problem: Main Problem 2})   only for restricted generalized collector paths of $G$ is equivalent to solving problem \ref{problem:min_mov} (\ref{problem:max_energy}) for $G$. 
\end{corollary}  
\begin{proof}
It follows directly from Lemmas \ref{lemma: vertical simplification} and \ref{Lemma: second energy simplification}.
\end{proof}

\begin{lemma}\label{lemma: simplification lemma}
    Let $G$ be a weighted grid graph associated to a SCA. Then, for every GCP $\PP$ there exists a restricted generalized collector path $\PP_G$ such that the number of  waypoints of $\PP_G$ is at most the number of waypoints of $\PP$ and $E(\PP) \le E(\PP_G)$. Moreover if every static point $P\in \SP$ satisfies $u_1 \le f(P)\le u_2$ for some $u_1, u_2$ real numbers, then every static point $P' \in \mathcal{S}_{\PP_G}$ satisfies $u_1 \le f(P') \le u_2$.
\end{lemma}
\begin{proof}
    We will find $\PP_G$ by applying a series of transformations to $\PP$ that will neither increase the number of waypoints in $\PP$ nor decrease the value of
    $f(P)$ at each static point.
    Let $\PP = \{P_i= (p_i, t_i)\}_{i = 1}^m$. \par
    \emph{First transformation:} If the last two points $P_{m-1}$ and $P_m$ satisfy $t_{m-1} = t_m = 180$, then remove the point $P_m$ from $\PP$. It is clear that this transformation satisfies the conditions of the Lemma.\par
    \emph{Second transformation:} For every pair of two consecutive points $P_i = (p_i, t_i)$ and $P_{i+1} = (p_{i+1}, t_{i+1})$ such that $t_{i+1} > t_i$, compute the maximum $x_j$ coordinate of the grid such that $x_j \le p_i = p_{i+1}$. Remove the points $(p_i, t_i)$ and $(p_{i+1}, t_{i+1})$ from $\PP$ and add the points $P_i'= (x_j, t_i)$ and $P_{i+1}'= (x_j, t_{i+1})$ in case such points were not previously on $\PP$. Clearly, this transformation would not increase the number of points and it will satisfy the conditions of the Lemma, as $f$ is constant in each cell of the grid, i.e. any point $P$ in the vertical segment from $P_i$ to $P_{i+1}$ would have a corresponding $P'$ in the border of the cell with $f(P) = f(P')$.\par 
    
    \emph{Third transformation:} The third transformation is oriented to remove unnecessary points of $\PP$ after the previous processing. Let $P_{i-1}, P_i, P_{i+1}$ three colinear points of $\PP$, then remove $P_i$. 
    
Notice that after these three transformations, all the points of the new $\PP$ are in vertical lines of the grid; moreover every vertical line that contains points does not have neither three consecutive points nor isolated points, i.e. all points in a line can be grouped in pairs of consecutive points. In addition, if a point $(p_i, t_i)$ of $\PP$ is on the graph, and the previous one $(p_{i-1}, t_{i-1})$ (or the following one) is at the same row, i.e $t_{i-1} = t_i$, then $(p_{i-1}, t_{i-1})$ is also on the graph. \par

\emph{Fourth transformation:} 
We will outline the key idea behind this transformation. The transformation proceeds sequentially in a backward direction, starting from the latest point in $\PP$ that does not belong to $G$.

After the previous transformations, we have that $(p_m, t_m)$ is a vertex of $G$ and $(p_{m-1}, t_{m-1})$ satisfies $p_m = p_{m-1}$ and $t_m > t_{m-1}$. Let $1 < i< m$ be the greatest  $i$ such that $P_i$ is not a vertex of $G$ but  $P_j \in G$ for every $j> i$. As noted above, we have $t_i< t_j$ for all $j> i$, and $p_i=p_{i+1}$. %In addition for $P_{i-1}$ it follows that $t_{i-1}=t_i$, i.e. $P_i$ and $P_{i-1}$ are on the same row because there are not three colinear points in $\PP$. 
Now, $P_i$ and $P_{i-1}$ are between two rows of $G$, say $k$ and $k+\varepsilon$. Take the sequence $\mathcal{R}=\{P_{i-l}, \dots, P_i\}\subset\PP$ of all points of $\PP$ that lie between such rows of $G$. We notice that this sequence defines a collection of vertical segments $ \mathcal{V} = \{v_1, v_2, \cdots, v_s\}$ for some $s$, such that the sum of the lengths of these segments is exactly $\varepsilon$ (see Figure \ref{fig:lemma5fourthtansformation}).

\begin{figure}
    \centering
        \includegraphics[width=0.8\linewidth]{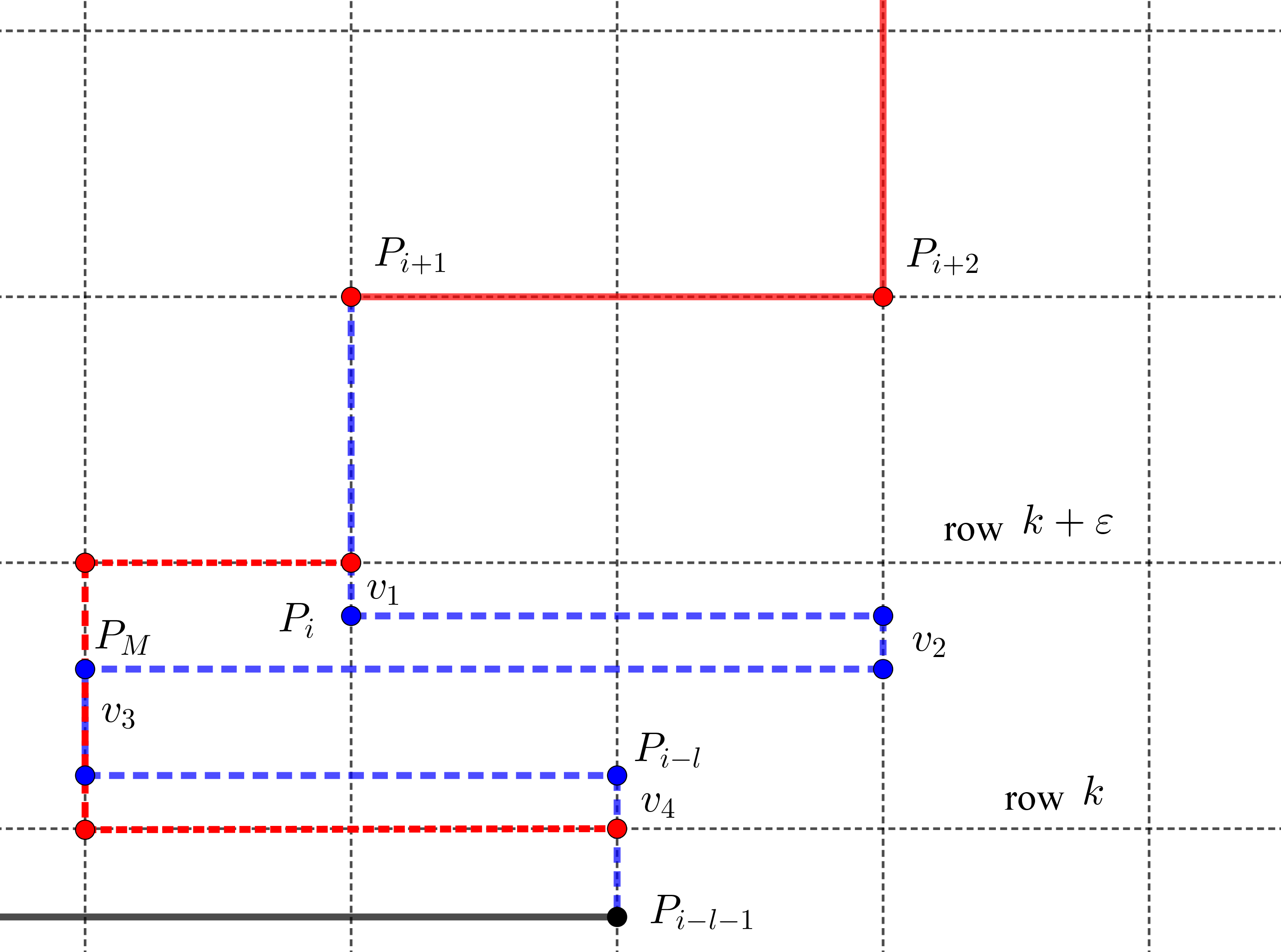}
        \caption{Example of fourth transformation. Blue points define the sequence $\mathcal{R}$. $v_3$ is the vertical segment with maximum irradiance function among $\{v_1, v_2, v_3, v_4\}$. Therefore we modify $\PP$ by removing all the blue points from $\PP$ and adding the new red ones.}
    \label{fig:lemma5fourthtansformation}
\end{figure}
 The idea of the transformation is to modify the sequence $\mathcal{R}=\{P_{i-l}, \dots, P_i\}\subset\PP$ in such a way we cross the rows $k+\varepsilon$, $k$ completely via the vertical segment $v_{\max} \in \mathcal{V}$, where $v_{\max}$ correspond to the segment with maximum irradiance function among all of the segments of $\mathcal{V}$. This transformation will not increase the number of waypoints in the path. Additionally, it will not decrease the value of the irradiance function at any point, thereby maintaining the desired energy collection. 
 There are various cases to consider based on the position of the vertical segment and whether
 $P_{i+1}$  %\st{lies within the} 
 is in the line $k+\varepsilon$ or not. 
 For clarity, we will present the formal analysis for one specific case, while the other cases are left for the reader to explore.

If $|R|=2$ then the modifications are trivial. In other case, $|\mathcal{R}|\ge 4$ as there are not three collinear points in $\PP$ (see the third transformation). Now,
let $P_M\in \mathcal{R}$ such that $f(P_M)=\max_{P_r \in \mathcal{R}} f(P_r)$. Notice that, if every static point $P \in S_{\PP}$ satisfies $u_1 \le f(P) \le u_2$, then  $u_1 \le f(P_M)\le u_2$. Assume that $P_M \ne P_i$ and $P_M \ne P_{i-l}$. Let $P_M=(p_M, t_M)$. Then remove $\mathcal{R}$ from $\PP$. If $k+\epsilon = t_{i+1}$, remove $(p_{i+1}, t_{i+1})$ from the path and instead add the new points $(p_M, t_{i+1})$, $(p_M, k)$ and $(p_{i-l}, k)$. If $k+\epsilon< t_{i+1}$,  add $(p_i, k+\epsilon)$, $(p_M, k+\epsilon)$ $(p_M, k)$ and $(p_{i-l}, k)$ (this last case is shown in Figure \ref{fig:lemma5fourthtansformation}). %Notice that the worst case is when $\mathcal{R}$ contains at least three waypoints (otherwise the modification is trivial). However, because there are not three collinear points in $\PP$ (see the third transformation), 
We have 
$|\mathcal{R}|\ge 4$ but we have added at most four points; hence, 
we have not increased  the number of waypoints in $\PP$, the value of the irradiance function has not decreased and the Lemma follows.
\end{proof}

\begin{theorem}
    Solving problems \ref{Problem: Main Problem 1} and \ref{Problem: Main Problem 2} %for a weighted grid graph
    is equivalent to solving problems \ref{problem:min_mov} and \ref{problem:max_energy}, respectively.
\end{theorem}
\begin{proof}

Clearly, any solution for problem \ref{problem:min_mov} or problem \ref{problem:max_energy} is feasible for problem \ref{Problem: Main Problem 1} or problem \ref{Problem: Main Problem 2}, respectively. On the contrary, by Lemma \ref{lemma: simplification lemma}, any solution for problem \ref{Problem: Main Problem 1} can be transformed into a feasible solution for  problem \ref{problem:min_mov}, or any solution for  problem \ref{Problem: Main Problem 2} can be transformed into a solution with no less energy collection for problem \ref{problem:max_energy}. 
\end{proof}

\section{Algorithms}
\label{sec:algorithms}

In this section, we present a polynomial-time solution for our target problems, parameterized by the number of vertices in the grid graph \( G \). Note that a restricted generalized collector path \( \PP_G \) induces a path in \( G \) by including all vertices of \( G \) that lie along the segments defined by \( \PP_G \). Let $\PP_G'$ be such a path. 
For simplicity, we will refer to vertices along $\PP_G'$
that belong to \( \PP_G \) %(excluding the initial and terminal vertices of \( \PP_G \))
as \textit{turn points}; see Figure \ref{fig:examplepath} on the right, which shows a path with eight turn points.
Consequently, analogous versions of the previous problems can be formulated in terms of minimizing the number of turn points (or equivalently, the number of turns)
of $\PP_G'$, rather than its cardinality.
We further classify each turn point \( P \) as a \textit{left turn point}, %or
a \textit{right turn point}, or a \textit{vertical turn point}, depending on whether the subsequent vertex lies to the left, %or 
to the right,
or in the same vertical line
of \( P \), respectively. From now on, with slight abuse of notation, we will refer to %this path
$\PP_G'$
as \( \PP_G \).
%}

In the next sections we first analyze the case when no left turns are allowed in the motion of the tracking system. This setting is inspired in practical applications, and serves as an introduction to our algorithms for solving the general case. We append the letters NL (No Left) to the acronyms of our target problems when analyzing the cases where no left turn are allowed. We show that 3D-MTM and 3D-MTM-NL (alternatively 3D-MEC and 3D-MEC-NL) have the same computational complexity. Finally, we set $O=(0,0)$ as the starting position of the system.

\subsection{3D-MTM Problem}

\subsubsection{No left turns}
\label{sssection:3d-mtm-nl}

\begin{definition}
\label{def:levels}
    The set of vertices of $G$ at distance $i$ from $O$ is defined as level $i$, and denoted as $L_i$.
\end{definition}

Recall that any solution to this problem is an ordered set of vertices of $G$. As no left turns are allowed, the following lemmas are straightforward:
\begin{lemma}
\label{lem:mtm_vertex_per_level}
    In any solution to 3D-MTM-NL at most one vertex per level is used.
\end{lemma}

\begin{lemma}
\label{lem:one_vertex_per_level_mtm}
    In any solution %$\PP$  
    $\PP_G$ to 3D-MTM-NL there is exactly one vertex of $L_i$ in $\PP_G$ for every $i\in\{0,\ldots,|\PP_G-1|\}$.
\end{lemma}

The above lemmas allow to segment a solution to 3D-MTM-NL by levels.  
For every vertex $v$ in level $i+1$ let $\texttt{h}_v, \texttt{v}_v$ indicates the minimum number of 
turns 
to reach $v$ using a horizontal or a vertical line, respectively. In addition, let $\bar{v}, v'$ be the horizontal and vertical neighbor of $v$ in $L_i$, respectively. Then:
\begin{equation}
\label{eq:hv_update}
    \texttt{h}_v = \min(\texttt{h}_{\bar{v}}, \texttt{v}_{\bar{v}}+1)
\end{equation}
\begin{equation}
\label{eq:vv_update}
    \texttt{v}_v = \begin{cases}
        \infty & u_1\nleq w_{v'}\nleq u_2 \\
        \min(\texttt{h}_{v'} + 1, \texttt{v}_{v'}) & \text{otherwise}
    \end{cases}
\end{equation}
As the base case, $\texttt{h}_O$ and $\texttt{v}_O$ are both zero.
The above formulas are readily satisfied and enable the use of dynamic programming.

From the restrictions of the problem, both values of $\texttt{h}_v, \texttt{v}_v$ are necessary to obtain an optimal path. The main reason is the different treatment to horizontal and vertical lines, as a vertex might not be the origin of a vertical segment, but it can be used for horizontal displacements. Horizontal displacements allows to connect vertices at column $i$ and vertices at column $j$, $i\ll j$, even when all vertices in between does not satisfy the threshold condition.

\begin{theorem}
\label{thm:3dmtmnl_time}
    3D-MTM-NL can be solved in $O(|V(G)|)$.
\end{theorem}
\begin{proof}
    Equations \ref{eq:hv_update} and \ref{eq:vv_update} provide the update formula for the values $\texttt{h}_v, \texttt{v}_v$ of every vertex $v\in V(G)$. Vertices are updated by levels starting with the point $O$, and the values of $\texttt{h}_v$ and $\texttt{v}_v$ are each updated exactly once in constant time. 
    Finally, all vertices in the last row of $G$, termed as $V_n$,
    can be checked to obtain $v^*$ such that $\min(\texttt{h}_{v^*}, \texttt{v}_{v^*})=\min\{\min(\texttt{h}_v, \texttt{v}_v)| v\in V_n\}$. 
    The resulting path that connects $O$ and $v^*$ is the optimal solution to the 3D-MTM-NL problem. Moreover, it can be clearly obtained in $O(|V(G)|)$ time.
\end{proof}

\subsubsection{The general case}

For the general case it must be considered that a vertex can also be reached from the left. However, 
when a vertex is reached horizontally from direction $d$, then the path must turn vertically or continue its movement in direction $d$, leading to the following remark:

\begin{remark}
\label{rmk:left_or_right_in_row}
A valid path traversing row $i<n$ of $G$ goes exclusively left, right, or neither before moving vertically.
\end{remark}

As in the previous case, we will consider $\overrightarrow{\texttt{h}_v}, \overleftarrow{\texttt{h}_v}, \texttt{v}_v$ indicating the minimum number of turns to reach vertex $v$ using a left-horizontal, a right-horizontal, and a vertical line respectively. In addition, let $V_j$ be the set of points in row $j>0$ of $G$, and consider the set of points $V_j^{'} \subseteq V_j$ %in row $j>0$ of $G$ 
that can be reached from row $j-1$. Then we can provide dynamic programming formulas for
$\overrightarrow{\texttt{h}_v}, \overleftarrow{\texttt{h}_v}$, and $\texttt{v}_v$ 
considering the vertices of % the previous and the current row 
$V_j'$ and $V_{j-1}$
as follows:

\begin{equation}
\label{eq:general_righthv_update}
    \overrightarrow{\texttt{h}}_{v_{i,j}} = \begin{cases}
        %\texttt{v}_{v_{i,j-1}} + 1
        \infty & i \le \min\{k | v_{k, j}\in V_j^{'}\} \\
        % \min(\texttt{v}_{v_{i,j-1}} + 1, \overrightarrow{\texttt{h}}_{v_{i-1,j}}) & \text{otherwise} \\
        \min(\texttt{v}_{v_{i-1,j}} + 1, \overrightarrow{\texttt{h}}_{v_{i-1,j}}) & \text{otherwise}
    \end{cases}
\end{equation}
\begin{equation}
\label{eq:general_lefthv_update}
    \overleftarrow{\texttt{h}}_{v_{i,j}} = \begin{cases}
        % \texttt{v}_{v_{i,j-1}} + 1 & i = \max\{k | v_{k, j}\in V_j^{'}\} \\
        % \min(\texttt{v}_{v_{i,j-1}} + 1, \overleftarrow{\texttt{h}}_{v_{i-1,j}}) & \text{otherwise} 
        \infty & i \ge \max\{k | v_{k, j}\in V_j^{'}\} \\
        \min(\texttt{v}_{v_{i+1,j}} + 1, \overleftarrow{\texttt{h}}_{v_{i+1,j}}) & \text{otherwise}
    \end{cases}
\end{equation}
\begin{equation}
\label{eq:general_vv_update}
    \texttt{v}_{v_{i,j}} = \begin{cases}
        \infty & u_1\nleq w_{v_{i,j-1}}\nleq u_2 \\
        \min(\min(\overrightarrow{\texttt{h}}_{v_{i,j-1}}, \overleftarrow{\texttt{h}}_{v_{i,j-1}})+1, \texttt{v}_{v_{i,j-1}}) & \text{otherwise}
    \end{cases}
\end{equation}

An important difference from the no-left turn case is that within every row our solution demands an ordered update. However, the order is by column index so no additional computational time is required. When updating the value of $\overrightarrow{\texttt{h}}_{v_{i,j}}$ we start by the left most vertex of row $j$ that can be reached from row $j-1$. This vertex cannot be reached from the right by any other vertex in row $j$ according to Remark \ref{rmk:left_or_right_in_row}. A similar idea is used when updating the value of $\overleftarrow{\texttt{h}}_{v_{i,j}}$ starting by the right most vertex in $V_j^{'}$. In the following, we provide the update rule for the first row. Notice that $\overleftarrow{\texttt{h}}_{v_{i,0}}$ is set to $\infty$ for all vertices in row $0$ because the movement of the SCA always starts from $(0,0)$.

\begin{equation}
    \overrightarrow{\texttt{h}}_{v_{i,0}} = 0 \hspace{4mm} \forall i\in [1,\ldots, n]
\end{equation}
\begin{equation}
    \texttt{v}_{v_{i,0}} = \begin{cases}
        0 & i = 0 \text{ and } u_1 \le w_{v_{0,0}} \le u_2\\
        1 & i > 0 \text{ and } u_1 \le w_{v_{i,0}} \le u_2 \\
        \infty & \text{otherwise}
    \end{cases} 
\end{equation}

The general case can be solved with a strategy similar to the no-left turn case. Hence, the ideas applied to demonstrate Theorem \ref{thm:3dmtmnl_time} can be readily modified to prove the following statement.

\begin{theorem}
    3D-MTM can be solved in $O(|V(G)|)$.
\end{theorem}

\subsection{3D-MEC Problem}

\subsubsection{No left turns}

Consider the set $S_{i, j, k}$ of all paths $\PP_G$ connecting vertex $O$ with vertex $v_{i,j}$ with exactly $k$ turns such that % the horizontal segment 
$\{v_{i-1, j}; v_{i, j}\}\subseteq \PP_G$, and define  $\texttt{H}[i, j, k]=\max(\sum_{p\in \vset{P}}w_p \text{, for all } \PP_G \in S_{i, j, k}$). 
Let $\texttt{V}[i,j,k]$ to be defined analogously, but %replacing the horizontal segment condition by
enforcing
$\{v_{i, j-1}; v_{i, j}\}\subseteq \mathcal{P}_G$. As no left turns are allowed, both tables can be filled with dynamic programming using similar ideas to the ones defined in Section \ref{sssection:3d-mtm-nl}. We distinguish different cases when filling the tables $\texttt{H}$ and $\texttt{V}$.

\begin{itemize}
    \item \textit{Elements at the left border}. The elements in the left most column of $G$ cannot be reached horizontally, hence $\texttt{H}$ is ill defined in this scenario. As we are solving a maximization problem, we set this value to $-\infty$ except for $O$. On the other hand, elements in this column can only be reached with a path coming straight from $O$. The update rule for $\texttt{H}$ and $\texttt{V}$ is the following:
\end{itemize}
\begin{equation}
    \texttt{H}[0, j, k] = \begin{cases}
        0 & j = 0 \text{ and } k = 0 \\
        -\infty & e.o.c
    \end{cases}
\end{equation}
\begin{equation}
    \texttt{V}[0, j, k] = \begin{cases}
        -\infty & k>0 \\
        0 & j = 0 \text{ and } k=0 \\ 
        \texttt{V}[0, j-1, 0] + w_{0, j-1} & j > 0 \text{ and } k=0
    \end{cases}
\end{equation}

\begin{itemize}
    \item \textit{Elements at the bottom border}. The elements in the bottom row of $G$ cannot be reached vertically, hence $\texttt{V}$ is ill defined in this scenario except for $O$. On the other hand, the elements in this row can only be reached with paths that makes no turns, and the horizontal displacements add no cost when coming from $O$.
\end{itemize}
\begin{equation}
    \texttt{V}[i, 0, k] = \begin{cases}
        0 & i = 0 \text{ and } k=0 \\
        -\infty & e.o.c
    \end{cases}
\end{equation}
\begin{equation}
    \texttt{H}[i, 0, k] = \begin{cases}
        -\infty & k > 0 \\
        0 & e.o.c
    \end{cases}
\end{equation}

\begin{itemize}
    \item \textit{Elements not in the bottom or left border}. This case consider both $i>0$ and $j>0$. Elements of $G$ at this position cannot be reached by a path that makes no turns. When applying a horizontal (vertical) displacement we check the value of the vertex at the previous column (row), which has been already computed. Being $k$ the number of turns when filling the tables, we keep the same value for $k$ for a sequence of two horizontal (vertical) segments, but use $k-1$ when the sequence is vertical-horizontal (horizontal-vertical).
\end{itemize}
\begin{equation}
    \texttt{H}[i, j, k] = \begin{cases}
        -\infty & i> 0, j > 0, k = 0 \\
        \max(\texttt{H}[i-1, j, k], \texttt{V}[i-1, j, k-1]) & i> 0, j > 0, k > 0
    \end{cases}
\end{equation}
\begin{equation}
    \texttt{V}[i, j, k] = \begin{cases}
        -\infty & i> 0, j > 0, k = 0 \\
        \max(\texttt{H}[i, j-1, k-1], \texttt{V}[i, j-1, k]) + w_{i, j-1} & i> 0, j > 0, k > 0
    \end{cases}
\end{equation}

The aforementioned cases cover all possible scenarios for vertices of $G$. Notice that the solution to the 3D-MEC-NL problem is associated to the last row of the grid. Moreover, 
the solution to the 3D-MEC-NL problem is $\max\{\texttt{V}[i, n, m] \text{ $|$ } i\in [1, \ldots, n]\}$, being $n$ the number of rows and columns of $G$, and $m$ the maximum number of turns. 
As is typical in dynamic programming, the path corresponding to the optimal solution
can be retrieved from the matrices $\texttt{H}$ and $\texttt{V}$. Then, the following theorem holds.

\begin{theorem}
    3D-MEC-NL can be solved in $O(|V(G)|m)$.
\end{theorem}

\subsubsection{The general case}

As for 3D-MTM problem, Remark \ref{rmk:left_or_right_in_row} is valid for the 3D-MEC problem: any horizontal portion of a path either goes left or right. As there are no cycles within a row, computing the optimal path reaching a vertex $v_{i,j}$ with $k$ turns can be divided into two matrices $\overrightarrow{\texttt{H}}[i,j,k]$ and $\overleftarrow{\texttt{H}}[i,j,k]$. The rules for updating $\overrightarrow{\texttt{H}}$ and $\overleftarrow{\texttt{H}}$ are similar to the ones defined in the previous case for $\texttt{H}$. Notice that:
\begin{itemize}
    \item $\overrightarrow{\texttt{H}}[i,j,k]$ depends on $\overrightarrow{\texttt{H}}[i-1,j,k]$ and $\texttt{V}[i-1,j,k-1]$.
    \item $\overleftarrow{\texttt{H}}[i,j,k]$ depends on $\overleftarrow{\texttt{H}}[i+1,j,k]$ and $\texttt{V}[i-1,j,k-1]$.
    \item $\texttt{V}[i,j,k]$ depends on $\overrightarrow{\texttt{H}}[i,j-1,k-1]$, $\overleftarrow{\texttt{H}}[i,j-1,k-1]$ and $\texttt{V}[i,j-1,k]$.
\end{itemize}
For $\overleftarrow{\texttt{H}}$ we must consider initializing vertices in the right border instead of vertices in the left border. In addition, for row 0 $\overleftarrow{\texttt{H}}$ is set to $-\infty$. With values of $\overrightarrow{\texttt{H}}$, $\overleftarrow{\texttt{H}}$ and $\texttt{V}$ for row 0 we can update $\texttt{V}$ for row 1. In general, the update process per row is: 1) update $\texttt{V}$, 2) update $\overrightarrow{\texttt{H}}$, 3) update $\overleftarrow{\texttt{H}}$. As the update process is per row, and $\texttt{V}$ depends on values from the previous row, $\texttt{V}$ is defined correctly. On the other hand, $\overrightarrow{\texttt{H}}$ depends on values from the previous column. As values for $\texttt{V}$ are already computed, and the values of $\overrightarrow{\texttt{H}}$ for column 0 are part of the base case, then $\overrightarrow{\texttt{H}}$ can be computed for any column in the current row. A similar idea is applied to the case of $\overleftarrow{\texttt{H}}$.

As we have shown, the general case is not more complex than the no left turn case in terms of worst computing time. With just one additional matrix, 3D-MEC problem can be solved using a similar approach than the one defined for 3D-MEC-NL. Therefore, the following theorem holds.

\begin{theorem}
    3D-MEC can be solved in $O(|V(G)|m)$.
\end{theorem}

\section{Experiments}
\label{sec:experiments}

We design a set of experiments to showcase the capabilities of our algorithms. First, we obtain the irradiance function with the use of raytracing software and public databases. Then we select two different geometries for the Solar Collector Element, one to account for perfect conditions, and other one simulating an SCE with a poor shape. The proposed solutions for the 3D-MTM and 3D-MEC problems are implemented in these scenarios using python 3.9.
As the computational complexity is similar but we are interested in realistic scenarios, the experiments described hereinafter only consider the case where no left turns are allowed.

The software used in this study to perform the raytracing simulations is Tonatiuh++, an open source Monte-Carlo ray tracer for modelling the light collection and concentration system of Solar Concentrating Thermal (CST) systems developed by \cite{TonatiuhPP}. The simulated SCE has a length of 12m, a diameter of 4m and a focus point of the parabola at the position (0, 1.5, 0). The material of the mirror has a reflectivity of 0.95 with a Gaussian distribution slope error of 2mrad. The shape of the sun was chosen as a Buie shape with a circumsolar ratio of 2\%. All boundary conditions are shown in Table \ref{tab:boundary_conditions}.

\begin{table}[]
    \centering
    \begin{tabular}{c c}
        \hline
        \multicolumn{2}{c}{\textbf{Geometry properties}} \\
        \hline
        Parabola length & 12m \\
        Parabola diameter & 4m \\
        Tube length & 12m \\
        Tube diameter & 0.1m \\
        Focus point & y=1.5m \\
        \hline
        \multicolumn{2}{c}{\textbf{Optical properties}} \\
        \hline
        Mirror reflectivity & 0.95 \\
        \(\sigma_{\text{opt}}\) & 2mrad \\
        Sun shape & Buie \\
        Circumsolar ratio & 2\% \\
        Number of rays & $10^6$ \\
        \hline
        \multicolumn{2}{c}{\textbf{Location}} \\
        \hline
        Longitude &  $37^\circ$ 24' 42'' \text{N} \\
        Latitude &   $6^\circ$ 0' 21'' \text{W} \\
        \hline
    \end{tabular}
    \caption{Specifications for the raytracing simulations}
    \label{tab:boundary_conditions}
\end{table}

To identify the minimum number of rays that need to be simulated for each experiment, a convergence study has been conducted. 
Three distinct design points from August 18, 2019, were selected to calculate the relative error of the raytracing results based on the number of simulated sun rays.
The reference values for calculating the relative errors are based on simulations conducted using 300 million sun rays.
The relative errors of all three design points are shown in Figure \ref{fig:relativeerror}. 

It can be seen that to reach an absolute value of the relative error of $ \leq 1\%$, a minimum of 1 million sun rays need to be simulated. 
It is important to note that for highly accurate raytracing simulations, a relative error of $ \leq 0.1\%$ is desirable. However, due to the large number of simulations required for this study, we opted for a slightly higher acceptable error of $ \leq 1\%$ to ensure reasonable run-times.

\begin{figure}
    \centering
    \includegraphics[width=1\textwidth]{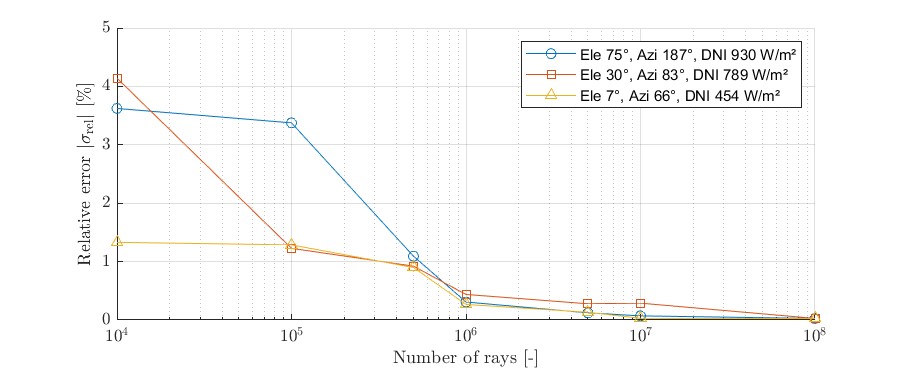}
    \caption{Relative error for different design points}
    \label{fig:relativeerror}
\end{figure}

Two possible mirror setups were investigated: 1) the mirrors composing the SCE have a perfect alignment, i.e. there are no errors in its parabolic shape; and 2) the mirrors are subject to minor displacements out of their optimal position. In the second setup the SCE is composed of 28 slightly distorted mirrors (Figure \ref{fig:picture_setup}) to consider possible misalignments in the SCE caused by the ageing of the components, and the operation of the solar field.

\begin{figure}
    \centering
    \includegraphics[width=.9\textwidth]{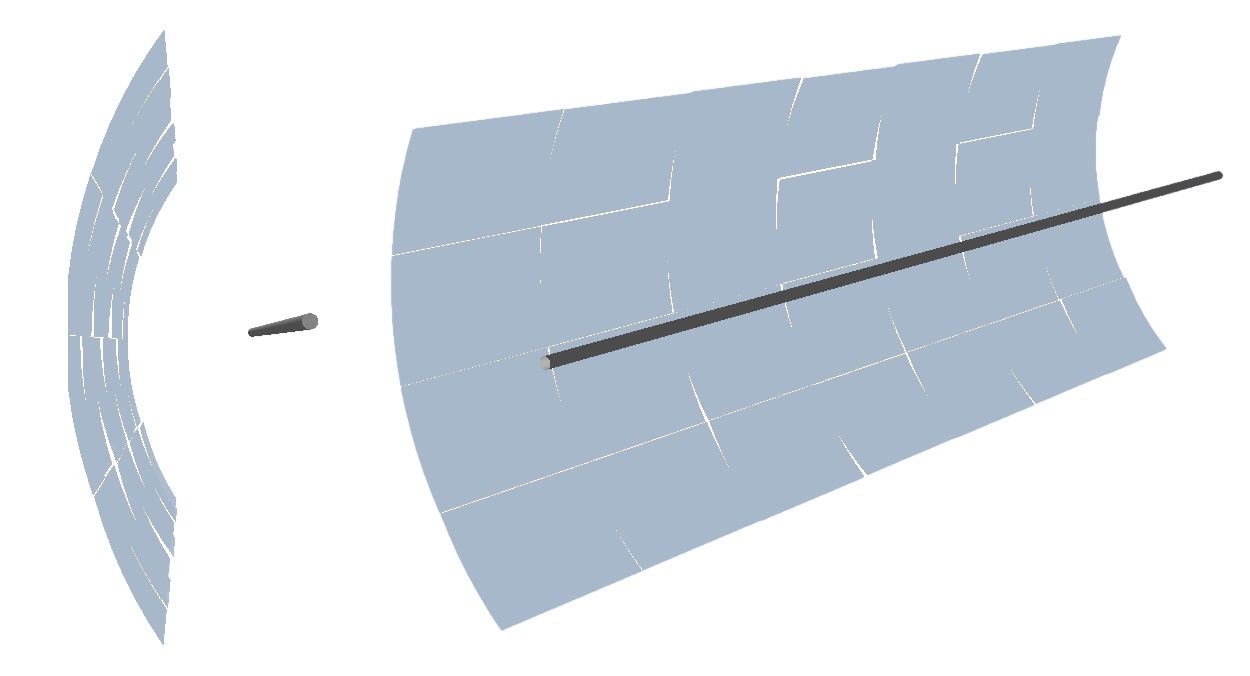}
    \caption{Geometry of setup 2) with displaced mirrors. Side view (left), front view (right)}
    \label{fig:picture_setup}
\end{figure}

The simulations are conducted with a fixed position of the SCE and the sun changing its position in 1-minute time steps. The angles of the SCE rotate between $[10\degree, 170\degree]$ % $[-80\degree, 80\degree]$
with a step of $\alpha=1^\circ$ for the case study 1) of the SCE with a perfect alignment. The case study 2) of displaced mirrors composing the SCE required smaller steps of changing the SCE angle to achieve accurate results and was therefore conducted with a step of $\alpha=0.2^\circ$. The whole range of motion of the SCE was passed for each minute of the selected days for this study. 

To account for realistic positions of the Sun while moving over the SCE, we use the information provided in \cite{SunPositionData}. This project provides the Direct Normal Irradiation (DNI) data for a given day with the selected time step. To find potential differences between cloudy and sunny days, we classify these two different cases through the daily DNI functions. By looking at the shape of the DNI function we can identify a sunny or cloudy day, see Figure \ref{fig:sunny_cloudy_example}. Notice we cannot classify sunny or cloudy days based on high DNI values because DNI changes throughout the year. We selected five sunny and five cloudy days corresponding to the year of 2019. 
All DNI values for all chosen days were collected through ground and satellite measurements that are provided in the National Solar Radiation Database %by the National Renewable Energy Laboratory 
\cite{IrradiancedataNREL}.

\begin{figure}
    \centering
    \begin{subfigure}{0.47\textwidth}
    % Answer: [trim={left bottom right top},clip]
        \includegraphics[width=\textwidth,trim={15cm 3cm 20cm 7cm},clip]{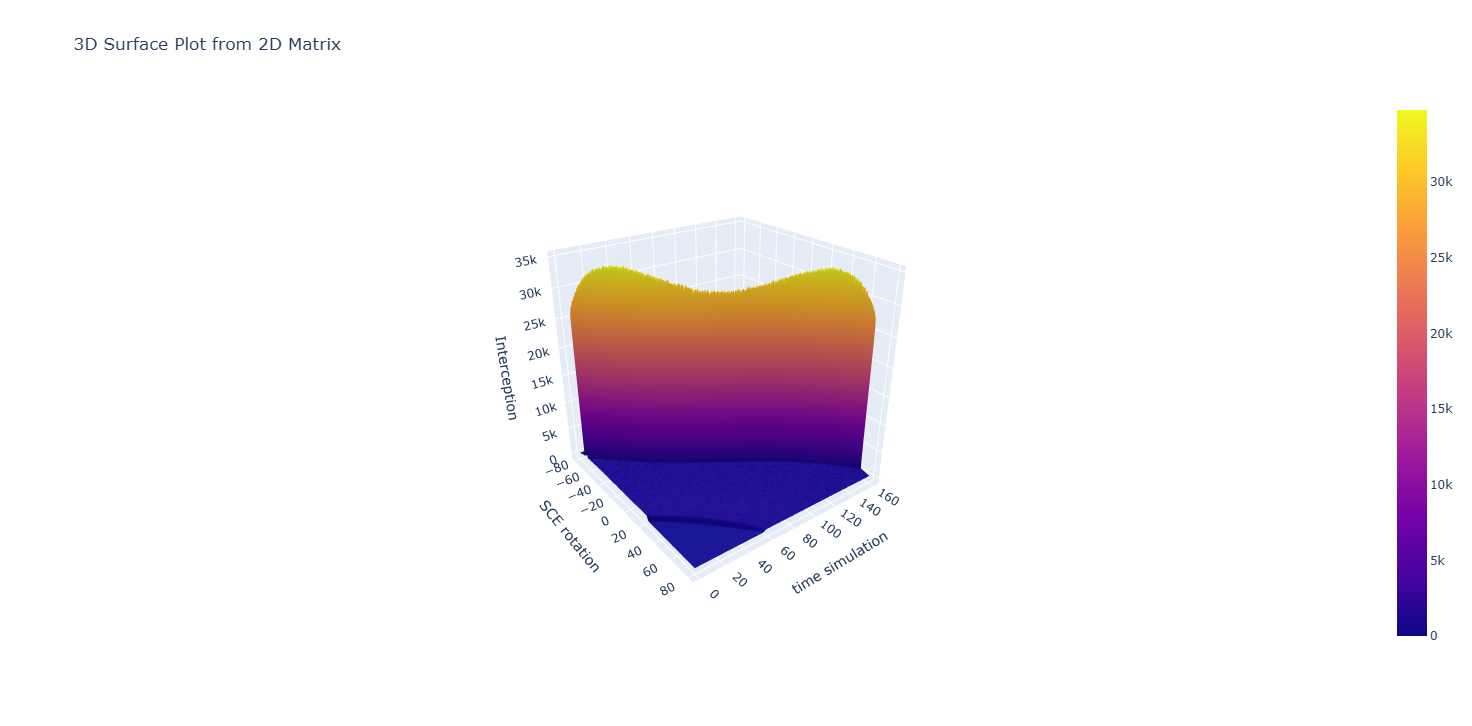}
        \caption{}
    \end{subfigure}
    \hfill
    \begin{subfigure}{0.47\textwidth}
        \includegraphics[width=\textwidth,trim={15cm 3cm 20cm 7cm},clip]{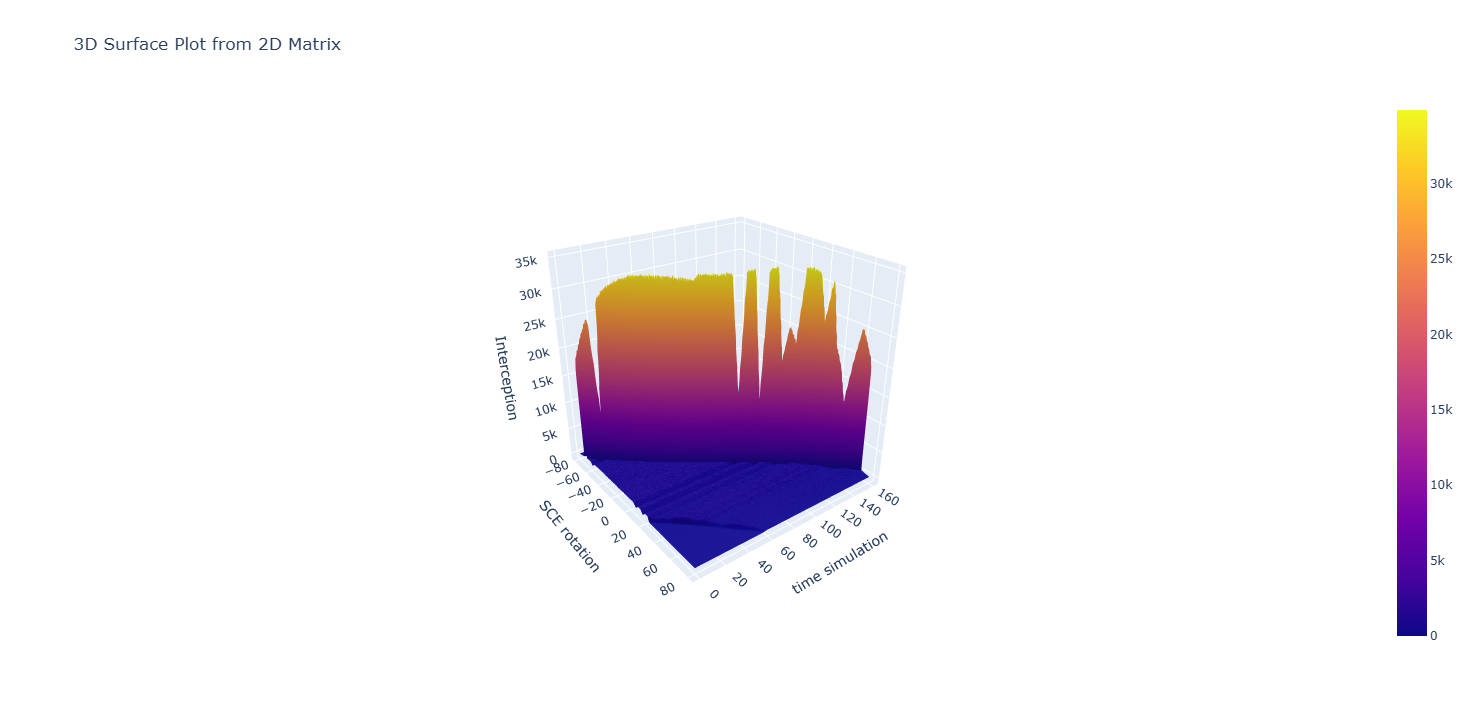}
        \caption{}
    \end{subfigure}
    \caption{Example of the irradiance function in (a) sunny, and (b) cloudy days. }
    \label{fig:sunny_cloudy_example}
\end{figure}

\subsection{3D-MTM}

The goal of this test is to observe the impact of changing the threshold for energy collection under different conditions. To this end, for any day $d$, we fix the value of $u_2$ as the maximum energy collected during $d$. Then we change the value of $u_1$ in the range of $[0, 15000]$ with a step of 200. Our results are depicted in Figure \ref{fig:3d_mtm_results}, from which we obtain several insights. We summarize them in the following bullet points:

\begin{figure}
    \centering
    \begin{subfigure}{\textwidth}
        \begin{subfigure}{0.49\textwidth}
    % Answer: [trim={left bottom right top},clip]
        \includegraphics[width=\textwidth,trim={0 0 1.3cm 1.4cm},clip, height=4.3cm]{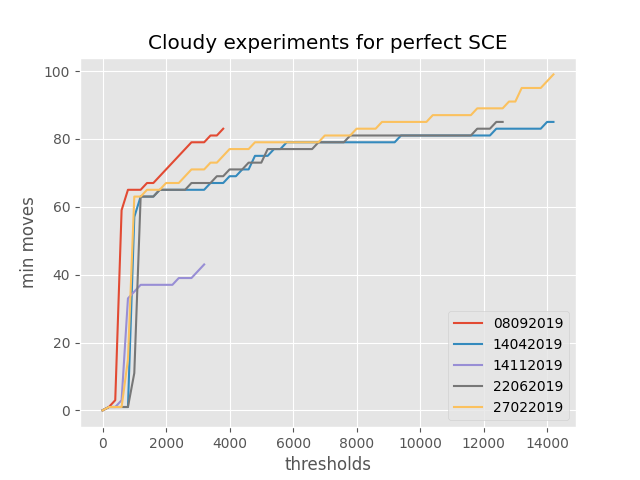}
    \end{subfigure}
    \hfill
    \begin{subfigure}{0.49\textwidth}
        \includegraphics[width=\textwidth,trim={0 0 1.3cm 1.4cm},clip, height=4.3cm]{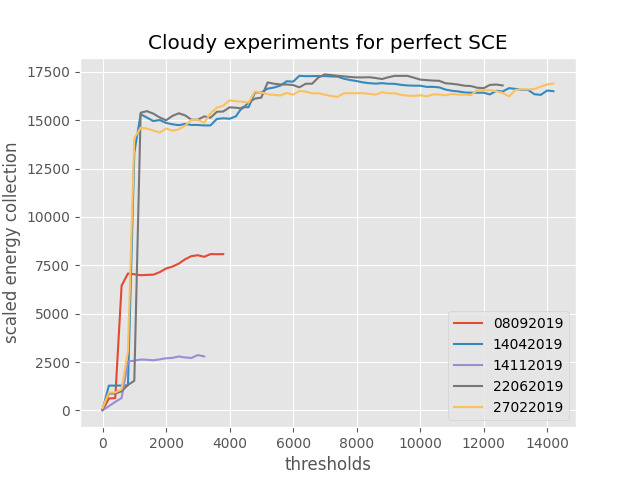}
    \end{subfigure}
    \caption{Cloudy experiments over perfect SCE}
    \end{subfigure}

    \begin{subfigure}{\textwidth}
        \begin{subfigure}{0.49\textwidth}
    % Answer: [trim={left bottom right top},clip]
        \includegraphics[width=\textwidth,trim={0 0 1.3cm 1.4cm},clip, height=4.3cm]{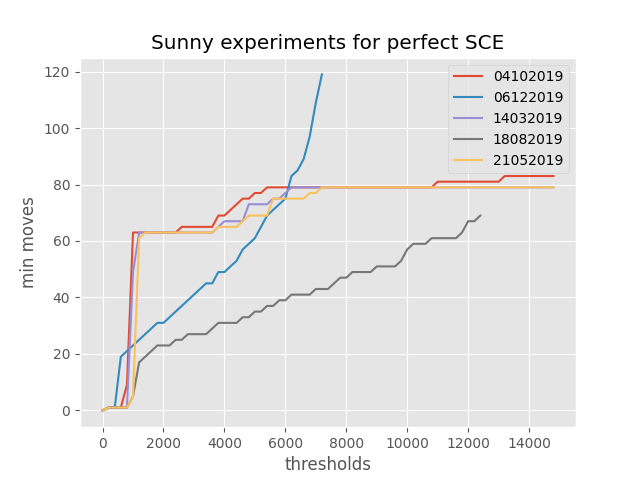}
    \end{subfigure}
    \hfill
    \begin{subfigure}{0.49\textwidth}
        \includegraphics[width=\textwidth,trim={0 0 1.3cm 1.4cm},clip, height=4.3cm]{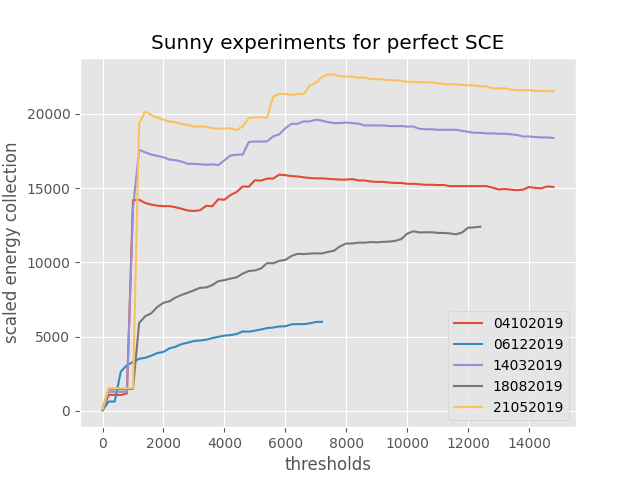}
    \end{subfigure}
    \caption{Sunny experiments over perfect SCE}
    \end{subfigure}
    
    \begin{subfigure}{\textwidth}
        \begin{subfigure}{0.49\textwidth}
    % Answer: [trim={left bottom right top},clip]
        \includegraphics[width=\textwidth,trim={0 0 1.3cm 1.4cm},clip, height=4.3cm]{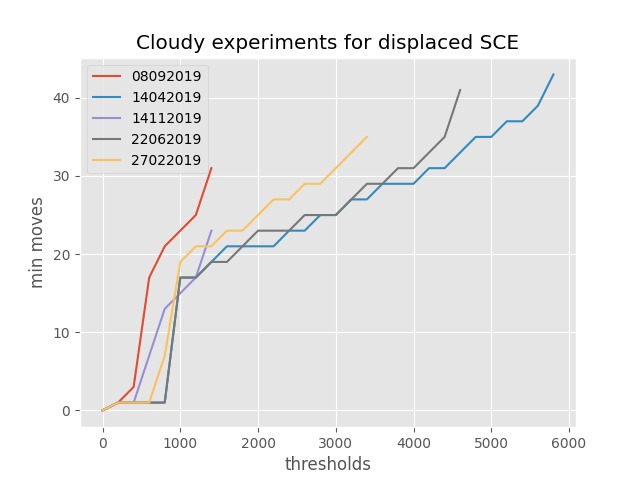}
    \end{subfigure}
    \hfill
    \begin{subfigure}{0.49\textwidth}
       \includegraphics[width=\textwidth,trim={0 0 1.3cm 1.4cm},clip, height=4.3cm]{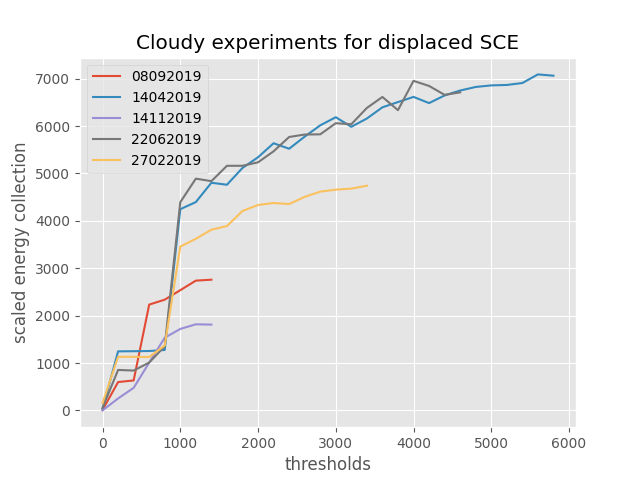}
    \end{subfigure}
    \caption{Cloudy experiments over displaced SCE}
    \end{subfigure}
    \begin{subfigure}{\textwidth}
        \begin{subfigure}{0.49\textwidth}
    % Answer: [trim={left bottom right top},clip]
        \includegraphics[width=\textwidth,trim={0 0 1.3cm 1.4cm},clip, height=4.3cm]{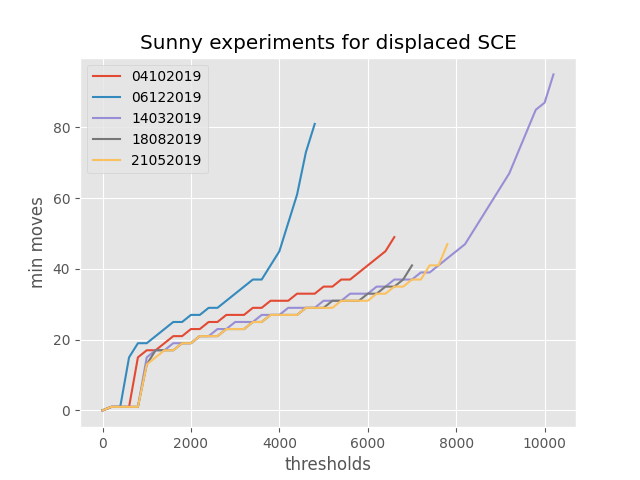}
    \end{subfigure}
    \hfill
    \begin{subfigure}{0.49\textwidth}
        \includegraphics[width=\textwidth,trim={0 0 1.3cm 1.4cm},clip, height=4.3cm]{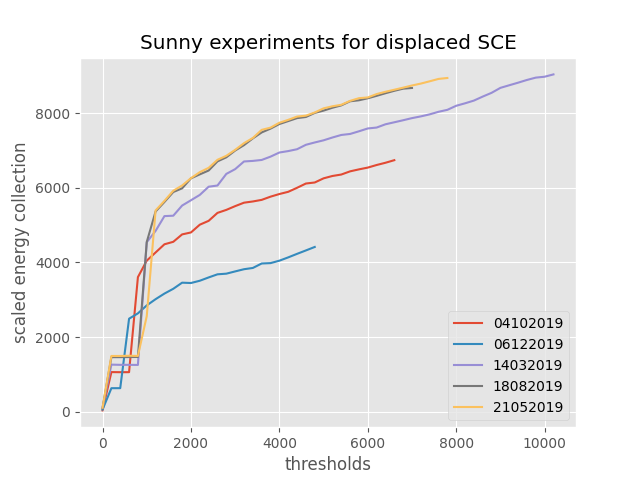}
    \end{subfigure}
    \caption{Sunny experiments over displaced SCE}
    \end{subfigure}
    \caption{3D-MTM results when changing the value of $u_1$. In every row the x-axis represent the threshold value while the $y$-axis represent the minimum number of movements of a R-GCP solving Problem \ref{problem:min_mov} (to the left), and the energy collected during the tracking (to the right). The energy value is divided by 1000 to reduce its range. 
    Each day $d$ is represented up to the maximum threshold such that the problem is solvable for the ray incidence function of $d$.
    }
    \label{fig:3d_mtm_results}
\end{figure}

\begin{itemize}
    \item The experiments conducted on the selected days of September, November (cloudy) and December (sunny) do not support energy capture thresholds above 4000, 3700, and 7800, respectively.  This implies that no feasible set of SCE rotations can ensure sufficient energy capture on these days. For the remaining experiments, energy capture levels at any time can exceed 12000.
    \item For experiments admitting high threshold values (above 12000) in Figures \ref{fig:3d_mtm_results}(a) and (b) we can observe a plateau in the graphics, excluding August 18, in Figure \ref{fig:3d_mtm_results}(b). This may suggest that the total number of rotational movements of the SCE can be reduced with minimal impact on total energy capture.
    \item 
    Experiments with displaced mirrors in the SCE show a lower energy capture, as expected, due to less efficient concentration of rays onto the absorber tube. Notably, neither energy collection nor the number of rotations reached a plateau in any of these trials, suggesting that energy capture efficiency might improve either by increasing rotational movements or by optimizing the timing of rotations. As demonstrated in the following section, energy capture can be significantly enhanced by a few additional movements performed at optimal times.
\end{itemize}

\subsection{3D-MEC}

The goal of this test is to observe the impact of changing the number of movements for energy collection under different conditions. When computing the irradiance function for the displaced HCE a higher number of displacement were considered; hence we set the maximum value of $m$ for this scenario in $360$, and $180$ for the perfect SCE. Our experiments demonstrate that is not required a higher range; see Figure \ref{fig:3d_mec_results}.  An interesting observation is that above 150 moves not much higher amount of energy is collected, independently of the evaluated scenario. This indicates that some rotational movements from the Solar Collector Assembly can be reduced, which might be beneficial to reduce the stress in the Ball Joint structure. We will conduct a more in-depth experiment in what follows, obtaining the exact number of rotations that can be reduced while capturing at least $95\%$ of the solar energy.

\begin{figure}[t]
    \centering
    \begin{subfigure}{.49\textwidth}
       \includegraphics[width=\textwidth,trim={0 0 1.3cm 1.4cm},clip, height=4.5cm]{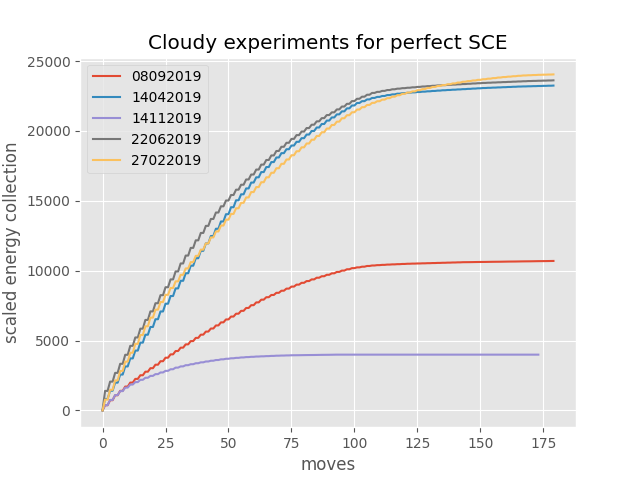}
    \caption{Cloudy experiments over perfect SCE}
    \end{subfigure}
    \hfill
    \begin{subfigure}{.49\textwidth}
       \includegraphics[width=\textwidth,trim={0 0 1.3cm 1.4cm},clip, height=4.5cm]{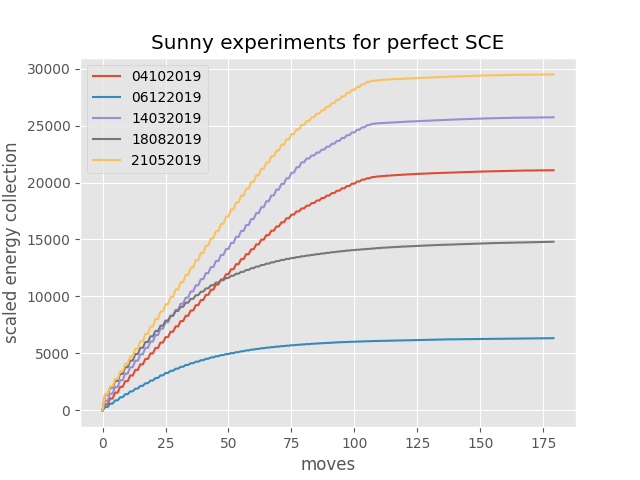}
    \caption{Sunny experiments over perfect SCE}
    \end{subfigure}

    \centering
    \begin{subfigure}{.49\textwidth}
       \includegraphics[width=\textwidth,trim={0 0 1.3cm 1.4cm},clip, height=4.5cm]{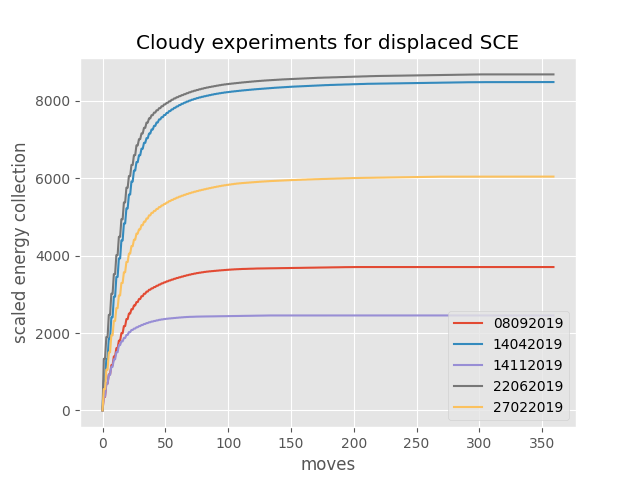}
    \caption{Cloudy experiments over displaced SCE}
    \end{subfigure}
    \hfill
    \begin{subfigure}{.49\textwidth}
       \includegraphics[width=\textwidth,trim={0 0 1.3cm 1.4cm},clip, height=4.5cm]{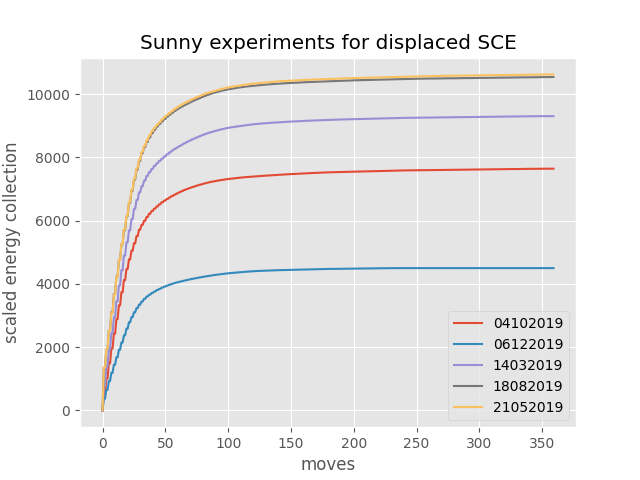}
    \caption{Sunny experiments over displaced SCE}
    \end{subfigure}
    \caption{3D-MEC results when increasing the number of movements across different days. $x$-axis represent the total number of movements of a R-GCP solving Problem \ref{problem:max_energy}, and $y$-axis the energy collected during the tracking period. The energy value is divided by 1000 to reduce its range.}
    \label{fig:3d_mec_results}
\end{figure}

\subsubsection{Adapting to short-time forecasting}

Our algorithms assume the irradiance function is known for the whole day. However, this is not the case when forecasting methods consider a smaller prediction range; see the work of \cite{kumari2021deep} for a comprehensive study. 

We are interesting in evaluating our methods assuming only a portion of the whole function is known. To this end, we design the following test: 1) given an irradiance function, split it into $k$ portions, 2) evaluate our 3D-MEC solution using $m$ steps in each of the $k$ intervals. We compare our results with the optimal solution for $m$ moves using the whole function as input. As the number of total movements is smaller for this case, it is expected to obtain a higher (but similar, according to Figure \ref{fig:3d_mec_results}) energy. We are also interested in computing the number of movements necessary to keep each of the $k$ intervals above the $95\%$ of the optimal solution for that interval. This is important to reduce the stress applied over the Ball Joint Assembly. Finally, we obtain the total computing time of our algorithms, which must also be considered when deciding the length of the forecasting period.

We consider the case where $k=2$ and $m=120$. In addition, the irradiance functions are scaled by dividing their values by 1000. For the perfect and displaced SCE, we consider the sunny and the cloudy experiments, respectively. Tables \ref{tab:fi_perfectSCE_sunny} and \ref{tab:fi_displacedSCE_cloudy} depict the results obtained in these scenarios.

Both for the perfect and the displaced SCE, it can be observed that using a shorter forecasting period results in higher energy collection at the cost of increasing the total number of rotations. However, it is remarkable that the computational effort is significantly reduced, as solving the problem twice for half a day each results in halving the overall processing time.

By examining the minimum number of rotations needed to collect at least $95\%$ of the solar energy within each forecasting interval, our results underscore the importance of solving our target problems. We compare the total number of movements across three scenarios: (1) the entire day ($m=120$), (2) a half-day forecasting interval, and (3) achieving 95\% energy collection within the half-day forecasting interval. Specifically, we calculate the percentage reduction in total movements in scenario 3 compared to scenarios 1 and 2. For an ideal SCE, the total number of rotations in scenario 3 represents a reduction of at least $44\%$ relative to scenario 2, and a reduction of at least $10\%$ relative to scenario 1. For a distorted SCE, the reduction in rotations is at least $63\%$ compared to scenario 2, and at least $27\%$ compared to scenario 1.

\begin{table}[t]
    \centering
    \begin{tabular}{|c|c|c|c|c|c|c|c|}
    \hline 
    \multirow{2}{2em}{Date} & \multicolumn{2}{c|}{Whole day} & \multicolumn{3}{c|}{$2$ FI}& \multicolumn{2}{c|}{$2$ FI (95\% energy)} \\ 
    \cline{2-8}
    & MEC & Time (s) & MEC & Moves & Time & $95\%$ Energy & Moves \\ \hline

    06122019 & 6145 & 10.17 & 6270 & 179 & 4.96 & 5979 & 99 \\ \hline

    18082019 & 14380 & 15.89 & 14815 & 239 & 7.40 & 14115 & 107 \\ \hline
    
    04102019 & 20690 & 11.83 & 20855 & 189 & 5.65 & 20020 & 101 \\ \hline

    14032019 & 25325 & 12.40 & 25445 & 181 & 6.02 & 24260 & 98 \\ \hline

    21052019 & 29110 & 15.11 & 29195 & 180 & 6.96 & 27925 & 97 \\ \hline

    \end{tabular}
    \caption{Results for the 3D-MEC problem applied over different sunny days, considering a perfect SCE and $m=120$. Columns related to the forecasting intervals are indicated with the acronym \textbf{FI}. The same value of $m\le120$ is considered on every forecasting interval.}
    \label{tab:fi_perfectSCE_sunny}
\end{table}

\begin{table}[t]
    \centering
    \begin{tabular}{|c|c|c|c|c|c|c|c|}
    \hline 
    \multirow{2}{2em}{Date} & \multicolumn{2}{c|}{Whole day} & \multicolumn{3}{c|}{$2$ FI}& \multicolumn{2}{c|}{$2$ FI (95\% energy)} \\ 
    \cline{2-8}
    & MEC & Time (s) & MEC & Moves & Time & $95\%$ Energy & Moves \\ \hline

    14112019 & 2447 & 70.00 & 2447 & 133 & 33.63 & 2344 & 45 \\ \hline

    08092019 & 3664 & 70.21 & 3698 & 197 & 33.74 & 3526 & 71 \\ \hline

    27022019 & 5897 & 69.48 & 5713 & 221 & 34.65 & 5442 & 75 \\ \hline

    14042019 & 8295 & 69.33 & 8435 & 239 & 33.84 & 8036 & 71 \\ \hline

    22062019 & 8502 & 80.07 & 8635 & 239 & 40.19 & 8226 & 69 \\ \hline

    \end{tabular}
    \caption{Results for the 3D-MEC problem applied over different cloudy days, considering a displaced SCE and $m=120$. Columns related to the forecasting intervals are indicated with the acronym \textbf{FI}. The same value of $m\le120$ is considered on every forecasting interval.}
    \label{tab:fi_displacedSCE_cloudy}
\end{table}

\section{Conclusions}
\label{sec:conclusions}

In this work, we addressed the challenge of reducing the rotational movements of Solar Collector Assemblies in Parabolic Trough solar fields while maintaining high levels of energy capture. Our research focused on two key problems: (1) minimizing the number of SCA rotations without compromising energy production and (2) maximizing energy capture when the number of allowable rotations is constrained. We provided a novel framework that transforms these continuous tracking challenges into path optimization problems on rectangular grids, allowing for polynomial-time solutions that are adaptable to varying weather conditions.

Our experimental results, based on realistic simulations and real-world data, demonstrate that SCA rotations can be reduced by at least 10\% while ensuring that energy collection remains above 95\% efficiency. This reduction in mechanical movements offers significant benefits, particularly in extending the operational lifespan of the solar tracking systems in PT plants, which are prone to wear and failure due to continuous use.

The robustness of our algorithms, which can integrate solar irradiance forecasting methods, ensures that our approach is well-suited for practical applications in diverse environmental conditions. These contributions mark a step forward in the sustainable operation of CSP plants by balancing energy efficiency with mechanical durability. Future research aims to extend these methods to different types of CSP technologies, consider different target functions, and explore more complex forecasting models for improved optimization.

\section*{Acknowledgments}

This work is partially supported by the grants PID2020-114154RB-I00, TED2021-129182B-I00, and DIN2020-011317 funded by MCIN/AEI/10.13039/501100011033 and the “European Union NextGenerationEU/PRTR".

\section*{Disclosure statement}

The authors declare that they have no known competing financial interests or personal relationships that could have appeared to influence the work reported in this paper.

\section*{Data availability statement}

Data is available on reasonable request from the authors.

\bibliographystyle{abbrv}
\bibliography{arxiv.bib}

\end{document}